\documentclass[a4paper, 10pt, conference]{ieeeconf} 

\IEEEoverridecommandlockouts                              
                                                         
\usepackage{xcolor}
\usepackage{cite}
\usepackage{amsmath,amssymb,amsfonts}
\usepackage{algorithmic}
\usepackage{algorithm}
\usepackage{graphicx}
\usepackage{subfig}
\usepackage{textcomp}
\usepackage{tikz}

\graphicspath{{pics/}}

\usepackage{mathtools}

\usepackage[acronym, style=alttree, toc=true, shortcuts, xindy, nomain, nonumberlist]{glossaries}

\usepackage{cleveref}

\usepackage{hyphenat}
\usepackage{bm}

\usepackage{siunitx}  

\usepackage{booktabs}

\usepackage{makecell}

\newcommand{\Mat}[1]{\bm{{#1}}}

\newcommand{\set}[1]{\mathcal{#1}} 
\newcommand{\setdef}[2][]{
	\left\{
		\ifblank{#1}{}{#1 \hspace{.1cm} \middle| \hspace{.1cm}}
		#2
	\right\}
} 

\newcommand{\lr}[1]{\left(#1\right)} 

\DeclareMathOperator\erf{erf} 

\DeclarePairedDelimiterXPP\onenorm[1]{}\lVert\rVert{_1}{\ifblank{#1}{\:\cdot\:}{#1}} 
\DeclarePairedDelimiterXPP\twonorm[1]{}\lVert\rVert{_2}{\ifblank{#1}{\:\cdot\:}{#1}} 

\newtheorem{theorem}{Theorem}
\newtheorem{assumption}{Assumption}

\newtheorem{definition}{Definition}

\newglossary{symbols}{sym}{sbl}{List of Symbols}
\newglossary{notation}{not}{nt}{Notation}

\glsaddkey{dimension}{\glsentrytext{\glslabel}}{\glsentryd}{\GLsentryd}{\glsd}{\Glsd}{\GLSd}
\glsaddkey{body}{\glsentrytext{\glslabel}}{\glsentrybody}{\GLsentrybody}{\glsbody}{\Glsbody}{\GLSbody}

\newacronym{MPC}{MPC}{Model Predictive Control}
\newacronym{MPCa}{MPC algorithm}{MPC algorithm}
\newacronym{RMPC}{RMPC}{Robust Model Predictive Control}
\newacronym{SMPC}{SMPC}{Stochastic Model Predictive Control}
\newacronym{SCMPC}{SCMPC}{Scenario Model Predictive Control}
\newacronym{ssc}{SCMPC}{Scenario Model Predictive Control}
\newacronym{MILP}{MILP}{Mixed Integer Linear Program}
\newacronym{PIT}{PIT}{Pointwise\hyp{}In\hyp{}Time}
\newacronym{POMDP}{POMDP}{Partially Observable Markov Decision Process}
\newacronym{MDP}{MDP}{Markov Decision Process}
\newacronym{KKT}{KKT}{Karush\hyp{}Kuhn\hyp{}Tucker}
\newacronym{EV}{EV}{ego vehicle}
\newacronym{TV}{TV}{target vehicle}
\newacronym{CA}{CA}{controlled agent}
\newacronym{DO}{DO}{dynamic obstacle}
\newacronym{cog}{CoG}{center of gravity}
\newacronym{ol}{OL}{Open-Loop}
\newacronym{cl}{CL}{Closed-Loop}
\newacronym{ocp}{OCP}{optimal control problem}
\newacronym{pog}{PG}{\textit{Probabilistic Grid}}
\newacronym{bog}{BG}{\textit{Binary Grid}}
\newacronym{LK}{LK}{lane keeping}
\newacronym{LC}{LC}{lane changing}
\newacronym{LCL}{LCL}{lane changing left}
\newacronym{LCR}{LCR}{lane changing right}
\newacronym{IA}{IA}{insignificant acceleration}
\newacronym{AC}{AC}{acceleration}
\newacronym{BR}{BR}{braking}
\newacronym{og}{OG}{Occupancy Grid}

\newglossaryentry{xik}{type=symbols,
	sort={variable},
	name={\ensuremath{\bm{\xi}_k}},
	description={...}
}
\newglossaryentry{nomxik}{type=symbols,
	sort={variable},
	name={\ensuremath{\overline{\bm{\xi}}_k}},
	description={...}
}
\newglossaryentry{uk}{type=symbols,
	sort={variable},
	name={\ensuremath{\bm{u}_k}},
	description={...}
}
\newglossaryentry{ek}{type=symbols,
	sort={variable},
	name={\ensuremath{\bm{e}_k}},
	description={...}
}
\newglossaryentry{xido}{type=symbols,
	sort={variable},
	name={\ensuremath{\bm{\xi}_k^{\text{DO}}}},
	description={...}
}
\newglossaryentry{xidoj}{type=symbols,
	sort={variable},
	name={\ensuremath{\bm{\xi}_k^{\text{DO},j}}},
	description={...}
}
\newglossaryentry{edo}{type=symbols,
	sort={variable},
	name={\ensuremath{\bm{e}^{\text{DO}}_k}},
	description={...}
}
\newglossaryentry{nomxido}{type=symbols,
	sort={variable},
	name={\ensuremath{\overline{\bm{\xi}}{}_k^{\text{DO}}}},
	description={...}
}
\newglossaryentry{nomxidoj}{type=symbols,
	sort={variable},
	name={\ensuremath{\overline{\bm{\xi}}{}_k^{\text{DO},j}}},
	description={...}
}
\newglossaryentry{nomxidoji}{type=symbols,
	sort={variable},
	name={\ensuremath{\overline{\bm{\xi}}{}_{k,i}^{\text{DO},j}}},
	description={...}
}
\newglossaryentry{xidos}{type=symbols,
	sort={variable},
	name={\ensuremath{\bm{\xi}_k^{\text{DO}}\left( s_i\right)}},
	description={...}
}
\newglossaryentry{xidosj}{type=symbols,
	sort={variable},
	name={\ensuremath{\bm{\xi}_k^{\text{DO},j}\left( s_i\right)}},
	description={...}
}
\newglossaryentry{xidoji}{type=symbols,
	sort={variable},
	name={\ensuremath{\bm{\xi}_{k,i}^{\text{DO},j}}},
	description={...}
}
\newglossaryentry{udo}{type=symbols,
	sort={variable},
	name={\ensuremath{\bm{u}_k^{\text{DO}}}},
	description={...}
}
\newglossaryentry{udoj}{type=symbols,
	sort={variable},
	name={\ensuremath{\bm{u}_k^{\text{DO},j}}},
	description={...}
}
\newglossaryentry{udot}{type=symbols,
	sort={variable},
	name={\ensuremath{\bm{u}_k^{\text{DO}}\left(T_i \right)}},
	description={...}
}
\newglossaryentry{udotj}{type=symbols,
	sort={variable},
	name={\ensuremath{\bm{u}_k^{\text{DO},j}\left(T^{j}_i \right)}},
	description={...}
}
\newglossaryentry{udos}{type=symbols,
	sort={variable},
	name={\ensuremath{\bm{u}_k^{\text{DO}}\left(s_i \right)}},
	description={...}
}
\newglossaryentry{udosj}{type=symbols,
	sort={variable},
	name={\ensuremath{\bm{u}_k^{\text{DO},j}\left(s^j_i \right)}},
	description={...}
}
\newglossaryentry{Udot}{type=symbols,
	sort={variable},
	name={\ensuremath{\bm{U}^{\text{DO}}\left(T_i \right)}},
	description={...}
}
\newglossaryentry{Udos}{type=symbols,
	sort={variable},
	name={\ensuremath{\bm{U}^{\text{DO}}\left(s_i \right)}},
	description={...}
}
\newglossaryentry{wk}{type=symbols,
	sort={variable},
	name={\ensuremath{\bm{w}_k}},
	description={...}
}
\newglossaryentry{wta}{type=symbols,
	sort={variable},
	name={\ensuremath{\bm{w}_k^{\text{ta}}}},
	description={...}
}
\newglossaryentry{wtaj}{type=symbols,
	sort={variable},
	name={\ensuremath{\bm{w}_k^{\text{ta},j}}},
	description={...}
}
\newglossaryentry{wex}{type=symbols,
	sort={variable},
	name={\ensuremath{\bm{w}_k^{\text{ex}}}},
	description={...}
}
\newglossaryentry{wexj}{type=symbols,
	sort={variable},
	name={\ensuremath{\bm{w}_k^{\text{ex},j}}},
	description={...}
}
\newglossaryentry{xisafe}{type=symbols,
	sort={variable},
	name={\ensuremath{\bm{\Xi}_k^{\text{safe}}}},
	description={...}
}
\newglossaryentry{xisafesi}{type=symbols,
	sort={variable},
	name={\ensuremath{\bm{\Xi}_k^{\text{safe}} \left( s_i \right)}},
	description={...}
}
\newglossaryentry{xisafes}{type=symbols,
	sort={variable},
	name={\ensuremath{\bm{\Xi}_k^{\text{safe}} \left( s \right)}},
	description={...}
}
\newglossaryentry{rparam}{type=symbols,
	sort={variable},
	name={\ensuremath{\beta}},
	description={...}
}
\newglossaryentry{rparamta}{type=symbols,
	sort={variable},
	name={\ensuremath{\gls{betata}}},
	description={...}
}
\newglossaryentry{rparamex}{type=symbols,
	sort={variable},
	name={\ensuremath{\gls{betaex}}},
	description={...}
}
\newglossaryentry{betata}{type=symbols,
	sort={variable},
	name={\ensuremath{\beta^{\text{ta}}}},
	description={...}
}
\newglossaryentry{betaex}{type=symbols,
	sort={variable},
	name={\ensuremath{\beta^{\text{ex}}}},
	description={...}
}
\newglossaryentry{PT}{type=symbols,
	sort={variable},
	name={\ensuremath{\set{P}_{\mathcal{T}}}},
	description={...}
}

\title{\LARGE \bf
Collision Avoidance with Stochastic Model Predictive Control for Systems with a Twofold Uncertainty Structure
}

\author{Tim~Br\"udigam$^{1}$,~Jie~Zhan$^{1}$,~Dirk~Wollherr$^{1}$,~and~Marion~Leibold$^{1}$
\thanks{$^{1}$T. Br\"udigam, J. Zhan, D. Wollherr, and M. Leibold are with the Chair of Automatic Control Engineering at the Technical University of Munich, Germany.
{\tt\small \{tim.bruedigam; jie.zhan; dw; marion.leibold\}@tum.de}}
}

\begin{document}

\maketitle
\thispagestyle{empty}
\pagestyle{empty}

\begin{abstract}
Model Predictive Control (MPC) has shown to be a successful method for many applications that require control. Especially in the presence of prediction uncertainty, various types of MPC offer robust or efficient control system behavior. For modeling, uncertainty is most often approximated in such a way that established MPC approaches are applicable for specific uncertainty types. However, for a number of applications, especially automated vehicles, uncertainty in predicting the future behavior of other agents is more suitably modeled by a twofold description: a high-level task uncertainty and a low-level execution uncertainty of individual tasks. In this work, we present an MPC framework that is capable of dealing with this twofold uncertainty. A scenario MPC approach considers the possibility of other agents performing one of multiple tasks, with an arbitrary probability distribution, while an analytic stochastic MPC method handles execution uncertainty within a specific task, based on a Gaussian distribution. Combining both approaches allows to efficiently handle the twofold uncertainty structure of many applications. Application of the proposed MPC method is demonstrated in an automated vehicle simulation study.
\end{abstract}

\section{Introduction}
\label{sec:introduction}

\vspace{-15cm}
\mbox{\small 
This~work~has~been~accepted~to~the~IEEE~2021~International~Conference~on~Intelligent~Transportation~Systems.}
\vspace{14.2cm}

Advances in research on automated systems are facilitating the use of controllers for complex applications, which is especially evident for automated vehicles. In many of these applications, there is one controlled agent, e.g., a vehicle or mobile robot, which is required to act and move among other agents. In order to move efficiently and avoid collisions, it is necessary for the controlled agent to anticipate the future behavior of the surrounding agents. 

The challenge here is that future behavior of other agents is subject to uncertainty. In many applications, this uncertainty consists of two types, task uncertainty and task execution uncertainty. 
Using automated vehicles as an example, the future motion of other surrounding vehicles is first subject to specific maneuvers, such as lane keeping or lane changing. Second, the execution of these maneuvers may vary again. A lane change may be executed quickly and aggressively, or slowly over a longer period of time.

\gls{MPC} is a suitable method to plan motion and trajectories for automated systems in environments with uncertainty. In \gls{MPC} an optimal control problem is solved on a finite horizon, utilizing prediction models to take into account the controlled agent dynamics and the future behavior of other agents. Constraints subject to environment uncertainty, e.g., for collision avoidance, may be handled robustly by using \gls{RMPC} methods for bounded uncertainties \cite{LangsonEtalMayne2004, RawlingsMayneDiehl2017}. However, these robust controllers are often highly conservative.

\gls{SMPC} approaches \cite{Mesbah2016, FarinaGiulioniScattolini2016} provide more efficient solutions compared to \gls{RMPC} by utilizing probabilistic chance constraints instead of hard constraints. These chance constraints enable increased efficiency by allowing a small probability of constraint violation, limited by a predefined acceptable risk. Various \gls{SMPC} methods exist, approximating the chance constraint to obtain a tractable representation that may be solved in an optimal control problem. In general, each \gls{SMPC} method considers one type of uncertainty within the prediction model.

Analytic \gls{SMPC} approaches \cite{SchwarmNikolaou1999, KouvaritakisEtalCheng2010, CarvalhoEtalBorrelli2014} yield an analytic approximation of the chance constraint, but these approaches are mostly restricted to Gaussian uncertainties. In particle-based \gls{SMPC} \cite{BlackmoreEtalWilliams2010} and \gls{SCMPC} \cite{SchildbachEtalMorari2014}, samples of the uncertainty are drawn that are then used to approximate the chance constraint. While arbitrary uncertainty distributions are possible, large numbers of samples are required to provide sufficient approximations for some uncertainty distributions, which increases computational complexity. If mixed uncertainty structures best describe the system behavior, the chance constraint approximations of these \gls{SMPC} approaches are not necessarily suitable. In \cite{BruedigamEtalWollherr2018b} an \gls{SMPC} framework, S+SC MPC, was introduced that utilizes both \gls{SCMPC} and a Gaussian uncertainty-based \gls{SMPC} method, specifically designed for a simple automated vehicle example.

In this paper, we propose an S+SC MPC framework that significantly generalizes the work of \cite{BruedigamEtalWollherr2018b}. In \cite{BruedigamEtalWollherr2018b} a simple S+SC MPC framework was specifically designed for automated vehicles, where only one surrounding vehicle and two possible maneuvers are considered. Here, we present a general S+SC MPC framework, applicable to a variety of automated systems. We specifically focus on collision avoidance, which requires considering multiple other agents that may perform multiple different tasks. 

The proposed S+SC MPC approach utilizes an \gls{SCMPC} approach for task uncertainty and an analytic \gls{SMPC} approach for task execution uncertainty. Combining these two approaches into a single \gls{MPC} optimal control problem allows to efficiently consider the twofold uncertainty structure of many practical applications with task and task execution uncertainty, e.g., automated vehicles \cite{CarvalhoEtalBorrelli2014, SchildbachBorrelli2015, CesariEtalBorrelli2017, BruedigamEtalLeibold2020c, MuraleedharanEtalSuzuki2020}. An automated vehicle simulation study illustrates the applicability of the proposed S+SC MPC framework.

The paper is structured as follows. Section~\ref{sec:problem} introduces the problem statement. The S+SC MPC method is derived in Section~\ref{sec:method}. A simulation study is presented in Section~\ref{sec:results}, followed by conclusive remarks in Section~\ref{sec:conclusion}.

\section{Problem Statement}
\label{sec:problem}

\gls{MPC} for collision avoidance with multiple agents requires two prediction models, one for the \gls{CA} and one for the \glspl{DO} to be avoided. 

We consider the CA dynamics
\begin{IEEEeqnarray}{c}
\label{eq:ca dynamics}
\bm{\xi}_{k+1} = \bm{f} \left( \gls{xik}, \gls{uk} \right) \label{eq:dynamics_CA}
\end{IEEEeqnarray}
depending on the 
nonlinear function $\bm{f}$ 
with state $\gls{xik}$ and input $\gls{uk}$ at time step $k$. 

Two types of uncertainties are considered for the DOs: task uncertainty and task execution uncertainty. This distinction reflects the situation of many applications, where the motion of surrounding agents is divided into discrete tasks with multiple task execution possibilities.

\begin{definition}[Tasks]
At each time step, a DO decides to execute exactly one task $T_i$ defined by the task set $\mathcal{T} = \setdef[T_i]{i = 1, ..., n_{\mathcal{T}}}$. Each task $T_i$ is assigned a probability $p_i$, subject to the probability distribution \gls{PT}, where $\sum_{i=1}^{n_{\mathcal{T}}} p_i = 1$ and $0 < p_1 \leq ... \leq p_{n_{\mathcal{T}}} \leq 1$. A DO input corresponding to task $T_i$ is denoted by $\bm{u}^{\text{DO}}(T_i)$.
\end{definition}

\begin{definition}[Task Execution]
Each task $T_i$ is subject to a nominal motion governed by the DO dynamics, a reference state, and an additive Gaussian uncertainty $ \gls{wex} \sim \mathcal{N} \left(\bm{0}, \bm{\Sigma}_k^{\text{ex}} \right)$ with covariance matrix $\bm{\Sigma}_k^{\text{ex}}$, representing uncertainty while executing task $T_i$.
\end{definition}

We consider multiple DOs. The dynamics for a single DO is then given by
\begin{IEEEeqnarray}{c}
\bm{\xi}_{k+1}^{\text{DO}} = \Mat{A}^{\text{DO}} \gls{xido} +\Mat{B}^{\text{DO}} \gls{udot} + \Mat{G}^{\text{DO}} \gls{wex} \label{eq:dynamics_DO}
\end{IEEEeqnarray}
with the DO state $\gls{xido}$, the input $\gls{uk}^{\text{DO}}$ as well as the state and input matrices $\Mat{A}^{\text{DO}}$, $\Mat{B}^{\text{DO}}$, $\Mat{G}^{\text{DO}}$. A DO stabilizing feedback controller is assumed of the form
\begin{IEEEeqnarray}{c}
\gls{udot} = \Mat{K}^{\text{DO}} \lr{ \gls{xido}  - \bm{\xi}^{\text{DO}}_{k,\text{ref}} \lr{T_i}}
\end{IEEEeqnarray}
with feedback matrix $\Mat{K}^{\text{DO}}$ and a reference state $\bm{\xi}^{\text{DO}}_{k,\text{ref}}$ depending on task $T_i$. The nominal state, assuming zero uncertainty and task $T_i$, follows
\begin{IEEEeqnarray}{c}
\overline{\bm{\xi}}{}_{k+1}^{\text{DO}} = \Mat{A}^{\text{DO}} \gls{nomxido} + \Mat{B}^{\text{DO}} \gls{udot}. \label{eq:dynamics_DO_nom}
\end{IEEEeqnarray}

Collisions with DOs are avoided by determining a set of safe states for the CA.

\begin{definition}
The safe set $\gls{xisafe}$ for time step $k$ ensures that all CA states $\gls{xik} \in \gls{xisafe}$ guarantee collision avoidance at time step $k$.
\end{definition}

We now formulate the \gls{ocp} to be solved within this work. Without loss of generality, the SMPC \gls{ocp} starts at time step $0$ where prediction steps are denoted by $k$. The SMPC \gls{ocp} is given by
\begin{IEEEeqnarray}{rl}
\IEEEyesnumber
J^* &= \min_{\bm{U}} J_N\left(\bm{\xi}_0, \bm{U} \right) \IEEEyessubnumber \label{eq:cost_setup}\\
\textrm{s.t. } & \bm{\xi}_{k+1} = \bm{f} \left( \gls{xik}, \gls{uk} \right)  \IEEEyessubnumber\\
&\bm{\xi}_{k+1}^{\text{DO}} = \Mat{A}^{\text{DO}} \gls{xido} +\Mat{B}^{\text{DO}} \gls{udot} + \Mat{G}^{\text{DO}} \gls{wex} \IEEEeqnarraynumspace \IEEEyessubnumber \label{eq:do_setup}\\
& \gls{uk} \in \set{U}, \hspace{25.5mm} k = 0, ..., N-1  \IEEEyessubnumber\\
& \gls{xik} \in \set{X}, \hspace{25.5mm} k = 1, ..., N  \IEEEyessubnumber\\
&\textrm{Pr}\left( \gls{xik} \in \gls{xisafe}  \right) \geq \gls{rparam}, \hspace{8mm} k = 1, ..., N \IEEEyessubnumber \label{eq:cc_setup}
\end{IEEEeqnarray}
with $\bm{U} = [\bm{u}_0, ..., \bm{u}_{N-1}]$, cost function $J_N$, horizon $N$, actuator constraints $\set{U}$, and deterministic state constraints $\set{X}$. As the DO dynamics~\eqref{eq:do_setup} are subject to uncertainty, the chance constraint \eqref{eq:cc_setup} is employed for collision avoidance. At each time step $k$, the probability of the CA state $\gls{xik}$ lying within the safe set $\gls{xisafe} $ must be larger than the risk parameter $\gls{rparam}= h \lr{\gls{betata}, \gls{betaex}}$, $0 \leq \beta \leq 1$. The function $h \lr{\gls{betata}, \gls{betaex}}$ indicates that $\beta$ depends on a task uncertainty risk parameter \gls{betata} and a task execution uncertainty risk parameter \gls{betaex}.

It is not possible to directly solve the chance-constrained \gls{ocp}. In the following, a method is derived that approximates the chance constraint \eqref{eq:cc_setup} to obtain a tractable \gls{ocp}. We first focus on task uncertainty in Section~\ref{sec:task_unc}, followed by task execution uncertainty in Section~\ref{sec:ex_unc}, which then allows to consider both uncertainties simultaneously as described in Section~\ref{sec:sscmpc}.

\section{Method}
\label{sec:method}

In the following, the S+SC MPC framework is derived, starting with individually focusing on \gls{SCMPC} and \gls{SMPC}.

\subsection{SCMPC for Task Uncertainty}
\label{sec:task_unc}

We first focus on task uncertainty. At each time step, one task is performed. The control action, corresponding to different tasks, may significantly vary between different tasks. Therefore, describing task uncertainty with Gaussian noise is impractical, rendering analytic SMPC approaches inapplicable. Considering every possible task may lead to highly conservative control behavior. However, applying SCMPC is a suitable approach to handle task uncertainty. With SCMPC, task uncertainty may be approximated by a small number of samples, as the number of possible tasks is usually small. In this section no task execution uncertainty is considered, i.e., $\gls{wex} = \bm{0}$.

Here, an SCMPC approach inspired by \cite{SchildbachEtalMorari2014} is used. By drawing $K$ samples from the probability distribution \gls{PT}, the task uncertainty is approximated, yielding the set of samples 
\begin{IEEEeqnarray}{c}
\mathcal{S} = \setdef[s_i]{i = 1, ..., K},
\end{IEEEeqnarray}
where a task $T_i$ is assigned to each sample $s_i$. An agent may execute the same task for multiple time steps. However, the agent task may change at every time step. 

\begin{assumption}
Within each SCMPC \gls{ocp}, each sampled task is assumed to be executed for the entire prediction horizon. \label{ass:onetask}
\end{assumption}

In other words, within the prediction, a sampled task is assumed to continue. This assumption is reasonable, as multiple tasks may be sampled and a new \gls{ocp} with new samples is initiated at each time step.

If Assumption~\ref{ass:onetask} holds, a DO input sequence is obtained for each sample of $\mathcal{S}$. The resulting input sequence $\gls{Udos} = \left[\bm{u}^{\text{DO}}_0 \left( s_i\right), ..., \bm{u}^{\text{DO}}_{N-1}\left( s_i\right)\right]$ depends on the individual inputs \gls{udot}, performing task $T_i$ corresponding to sample $s_i$. Based on \gls{Udos}, the predicted DO states for each sample are obtained according to the DO dynamics~\eqref{eq:dynamics_DO} with $\gls{wex} = \bm{0}$, resulting in the predicted states \gls{xidos} for $k = 1, ..., N$.

Depending on the predicted DO states, a safe set \gls{xisafesi} may be computed for each drawn sample $s_i$. Each safe set requires an individual constraint in the SCMPC \gls{ocp}. Therefore, for the SCMPC approach, the chance constraint~\eqref{eq:cc_setup} is adapted to
\begin{IEEEeqnarray}{c}
\textrm{Pr}\left( \gls{xik} \in \gls{xisafes}  \right) \geq \gls{rparamta}, ~~k = 1, ..., N, ~~ s\in \set{S}. \IEEEeqnarraynumspace \label{eq:cc_scmpc}
\end{IEEEeqnarray}
Multiple methods exist to generate safe sets, e.g., signed distance \cite{SchulmanEtalAbbeel2013} or grid-based methods \cite{BruedigamEtalLeibold2020c}.

The sample size $K$ depends on the chosen risk parameter. We propose a strategy to obtain $K$ that focuses on the least likely task $T_1$ in $\set{T}$.

\begin{theorem} \label{lem:samplesize}
The sample size
\begin{IEEEeqnarray}{c}
K > \log_{1-p_1} \left( \frac{1-\gls{betata}}{p_1}\right) \label{eq:samplesize}
\end{IEEEeqnarray}
ensures that the probability of not having sampled the least probable task $T_1$, if it later occurs, is lower than the allowed risk $1-\gls{betata}$, i.e., \eqref{eq:cc_scmpc} is satisfied.
\end{theorem}

\begin{proof}
The proof is based on \cite{BruedigamEtalWollherr2018b}. Given independent and identically distributed samples, the worst-case probability of not sampling task $T_1$, if it later occurs, is given by $p_1(1-p_1)^K$. The sample size $K$ in \eqref{eq:samplesize} then follows from solving for $K$ with $1-\gls{betata} > p_1(1-p_1)^K$, i.e., bounding the worst-case probability given the risk parameter \gls{betata}.
\end{proof}

If the least likely task $T_1$ is actually performed by the DO, this worst-case probability of not having sampled task $T_1$ is lower than the acceptable risk, defined by the SCMPC risk parameter \gls{betata}.

After having introduced an SCMPC approach to handle task uncertainty, the following section introduces an analytic SMPC approximation for task execution uncertainty.

\subsection{SMPC Task Execution Uncertainty}
\label{sec:ex_unc}

We now focus on task execution uncertainty, assuming only one task is possible. In the DO dynamics~\eqref{eq:dynamics_DO}, task execution uncertainty is described by the additive Gaussian uncertainty, representing uncertainty considering the nominal trajectory of a task. Approximating a Gaussian distribution potentially requires a large number of samples, therefore, an analytic SMPC approach is more suitable than SCMPC. The cost function~\eqref{eq:cost_setup} may depend on the DO uncertainty. Therefore, the cost is adjusted based on the expectation value, yielding
\begin{IEEEeqnarray}{c}
J_N = \mathrm{E} \lr{ \sum_{k=0}^{N-1} l\lr{\gls{xik}, \gls{uk}, \gls{wex}} + J_{\text{f}} \lr{\bm{\xi}_N, \bm{w}^\text{ex}_N} }
\end{IEEEeqnarray}
with stage cost $l$ and terminal cost $J_{\text{f}}$.

The constraint $\gls{xik} \in \gls{xisafe}$ may be described by a set of functions 
\begin{IEEEeqnarray}{c}
\bm{d}_k \lr{\gls{xik}, \gls{xido}} \geq \bm{0}  ~~ 	\Leftrightarrow ~~ \gls{xik} \in \gls{xisafe} \label{eq:cc_functions}
\end{IEEEeqnarray}
with $\bm{d}_k = [d_{k,1}, ..., d_{k,n_\text{d}}]^\top$, where $n_\text{d}$ denotes the number of constraint functions. 

In order to find an analytic approximation for the chance constraint~\eqref{eq:cc_setup} with only one task, a linearized description of the chance constraint is required. Therefore, the nonlinear constraint~\eqref{eq:cc_functions} is linearized around the nominal states with $\gls{xido} = \gls{nomxido} + \gls{edo}$ and the prediction error \gls{edo}. Based on \gls{wex}, the prediction error follows $\gls{edo} \sim \mathcal{N}\lr{\bm{0}, \bm{\Sigma}^{\text{e}}_k} $ where
\begin{IEEEeqnarray}{c}
\bm{\Sigma}^{\text{e}}_{k+1} = \bm{\Phi} \bm{\Sigma}_k^{\text{e}} \bm{\Phi}^\top + \Mat{G}^{\text{DO}} \bm{\Sigma}_k^{\text{ex}} {\Mat{G}^{\text{DO}}}^\top  \label{eq:propagation}
\end{IEEEeqnarray}
with $\bm{\Phi} = \Mat{A}^{\text{DO}} + \Mat{B}^{\text{DO}} \Mat{K}^{\text{DO}}$.

The resulting linearized description of~\eqref{eq:cc_functions} is
\begin{IEEEeqnarray}{c}
\bm{d}_k \lr{\gls{xik}, \gls{nomxido}} + \nabla \bm{d}^{\text{DO}}_k \gls{edo}  \geq \bm{0} \IEEEeqnarraynumspace  \label{eq:cc_functions_lin}
\end{IEEEeqnarray}
with
\begin{IEEEeqnarray}{c}
\nabla \bm{d}^{\text{DO}}_k =  \left. \frac{\partial \bm{d}_k }{\partial \gls{xido} } \right|_{\gls{xik}, \gls{nomxido}}. \IEEEeqnarraynumspace \label{eq:nabla}
\end{IEEEeqnarray}

The linearized chance constraint is then given by
\begin{IEEEeqnarray}{c}
\textrm{Pr}\left(\nabla \bm{d}^{\text{DO}}_k \gls{edo}  \geq - \bm{d}_k \lr{\gls{xik}, \gls{nomxido}} \right) \geq \gls{rparamex}, \IEEEeqnarraynumspace \label{eq:cc_lin}
\end{IEEEeqnarray}
which is still a probabilistic expression. However, \eqref{eq:cc_lin} may be approximated into an analytic expression similar to \cite{BruedigamEtalWollherr2018b}.

\begin{theorem}\label{lem:cc_analytic}
The probabilistic chance constraint \eqref{eq:cc_lin} may be approximated by the analytic expression
\begin{IEEEeqnarray}{l}
\IEEEyesnumber \label{eq:th2}
d_{k,i} \lr{\gls{xik}, \gls{nomxido}} \geq \bm{\gamma}_{k,i} \IEEEyessubnumber\\
\gamma_{k,i} = \sqrt{2 \nabla d^{\text{DO}}_{k,i} \bm{\Sigma}^{\text{e}}_k 
{\nabla d^{\text{DO}}_{k,i}}^\top}\erf^{-1}\lr{1-2\gls{betaex}}  \IEEEyessubnumber \IEEEeqnarraynumspace
\end{IEEEeqnarray}
with $\bm{\gamma}_k = [\gamma_{k,1}, ..., \gamma_{k,n_\text{d}}]^\top$ and $0.5 \leq \gls{betaex} \leq 1$.
\end{theorem}

\begin{proof}
The proof follows \cite{CarvalhoEtalBorrelli2014, BruedigamEtalWollherr2018b}. Due to~\eqref{eq:propagation} it holds that $\nabla \bm{d}^{\text{DO}}_k \gls{edo} \sim \mathcal{N}\lr{\bm{0},\nabla d^{\text{DO}}_{k,i} \bm{\Sigma}^{\text{e}}_k {\nabla d^{\text{DO}}_{k,i}}^\top}$ in~\eqref{eq:cc_lin}. The quantile function for univariate normal distributions allows to reformulate~\eqref{eq:cc_lin} into~\eqref{eq:th2}.
\end{proof}

Note that $ \nabla d^{\text{DO}}_{k,i}$ is defined similar to~\eqref{eq:nabla}. The individual approaches for handling task uncertainty and task execution uncertainty are combined in the following section.

\subsection{S+SC MPC Algorithm}
\label{sec:sscmpc}

The results of Section~\ref{sec:task_unc} and Section~\ref{sec:ex_unc} are now combined in order to obtain the S+SC MPC framework, which is able to efficiently handle the mixed uncertainty structure. In addition, multiple DOs are considered with the DO dynamics
\begin{IEEEeqnarray}{c}
\bm{\xi}_{k+1}^{\text{DO},j} = \Mat{A}^{\text{DO},j} \gls{xidoj} + \Mat{B}^{\text{DO},j} \gls{udotj}  + \Mat{G}^{\text{DO},j} \gls{wexj} \label{eq:dynamics_DOs} \IEEEeqnarraynumspace
\end{IEEEeqnarray}
with stabilizing feedback matrix $\Mat{K}^{\text{DO},j}$ for the DOs $j = 1, ..., n_{\text{DO}}$.

The tractable S+SC MPC \gls{ocp} for multiple DOs is then given by
\begin{IEEEeqnarray}{rl}
\IEEEyesnumber \label{eq:ocp_sscmpc}
J^* &= \min_{\bm{U}}  \mathrm{E} \lr{ \sum_{k=0}^{N-1} l\lr{\gls{xik}, \gls{uk}, \gls{wex}} + J_{\text{f}} \lr{\bm{\xi}_N, \bm{w}^\text{ex}_N} } \IEEEyessubnumber \IEEEeqnarraynumspace \label{eq:cost_ssc}\\
\textrm{s.t. } & \bm{\xi}_{k+1} =  \bm{f} \left( \gls{xik}, \gls{uk} \right)  \IEEEyessubnumber\\
&\overline{\bm{\xi}}_{k+1, i}^{\text{DO},j} = \Mat{A}^{\text{DO},j} \gls{nomxidoji} + \Mat{B}^{\text{DO},j} \gls{udosj}  \IEEEeqnarraynumspace \IEEEyessubnumber\\
& \gls{uk} \in \set{U}, \hspace{29,5mm} k = 0, ..., N-1  \IEEEyessubnumber\\
& \gls{xik} \in \set{X}, \hspace{29.5mm} k = 1, ..., N  \IEEEyessubnumber\\
&d^j_{k,i} \lr{\gls{xik}, \gls{nomxidoji}} \geq \bm{\gamma}_{k,i}^j, \hspace{7mm} k = 1, ..., N \IEEEyessubnumber \label{eq:cc_sscmpc}\\
&\gamma_{k,i}^j = \sqrt{2  \nabla d^{\text{DO},j}_{k,i} \bm{\Sigma}^{\text{e},j}_k 
{\nabla d^{\text{DO},j}_{k,i}}^\top }\erf^{-1}\lr{1-2\gls{betaex}} \IEEEyessubnumber \IEEEeqnarraynumspace 
\end{IEEEeqnarray}
with $i =  1, ..., K_j$ where $K_j$ is determined according to~\eqref{eq:samplesize} for each DO, given the DOs $j = 1, ..., n_{\text{DO}}$.

In~\eqref{eq:cc_sscmpc}, an individual approximated chance constraint is generated for each sample $s_i$, depending on $K_j$. While this approach is reasonable for a small number of samples, it becomes computationally expensive for a larger $K_j$. A possible alternative for application is to combine similar individual task in order to reduce the number of total constraints. This approach is illustrated in the simulation example in Section~\ref{sec:results}.

If it is required to guarantee safety or recursive feasibility, the proposed S+SC MPC method may be extended by the safety framework for \gls{SMPC} approaches proposed in~\cite{BruedigamEtalLeibold2021b}.

\section{Simulation Study}
\label{sec:results}

To evaluate the effectiveness of the S+SC MPC algorithm presented in Section~\ref{sec:sscmpc}, a highway scenario involving five \glspl{TV} is simulated, using the Control Toolbox~\cite{GiftthalerEtalBuchli2018}. Here, the \gls{CA} and \glspl{DO} become \gls{EV} and \glspl{TV}, respectively. The initial vehicle configuration is depicted in Figure~\ref{fig:init position scenario 5-tv}.
\begin{figure}
\vspace{1mm}
\centering
\includegraphics[width = 0.8\columnwidth]{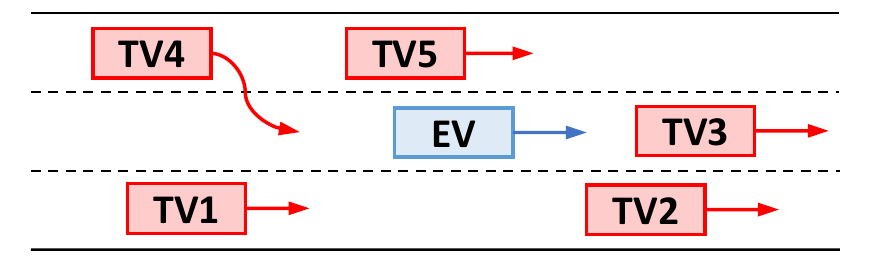}
\caption{Initial scenario configuration.}
\label{fig:init position scenario 5-tv}
\vspace{-4.5mm}
\end{figure}

We first present the results of the simulation study with the proposed S+SC MPC algorithm, and then, for comparison, we investigate the stand-alone algorithms \gls{SMPC} and \gls{SCMPC}. Eventually, we investigate applying S+SC MPC to varying scenario configurations.

\subsection{Simulation Setup}

All simulations are run on an Intel i5-2500K CPU @ 3.30GHz with 15.6GB RAM. Each simulation consists of $n_{\text{iter}}=100$ \gls{MPC} iterations, which is equivalent to a scenario duration of $ \SI{20}{\second}$ with $\Delta t = \SI{0.2}{\second} $. In the following, SI units are assumed for variables and parameters expressed without units.

As a special case of \eqref{eq:ca dynamics}, the \gls{EV} dynamics are represented using the linear, discrete-time point mass model
\begin{equation}
\label{eq:ev}
\bm{\xi}_{k+1}^{\text{EV}} = \Mat{A}  \bm{\xi}_{k}^{\text{EV}} + \Mat{B}  \bm{u}_{k}^{\text{EV}}
\end{equation}
with the \gls{EV} states ${\bm{\xi}_{k} = \left[ x_k, v_{x,k}, y_k, v_{y,k}\right]^\top}$ and inputs ${\bm{u}_{k} = \left[ u_{x,k}, u_{y,k}\right]^\top}$ where
\begin{equation}
\label{eq:AB}
\Mat{A} = \begin{bmatrix}
1 & \Delta t & 0 & 0 \\ 0 & 1 & 0 & 0 \\ 0 & 0 & 1 &  \Delta t \\ 0 & 0 & 0 & 1
\end{bmatrix},
\Mat{B} = \begin{bmatrix}
0.5\Delta t^2 & 0 \\  \Delta t & 0 \\ 0 & 0.5\Delta t^2 \\ 0 &  \Delta t
\end{bmatrix}.
\end{equation}

The \gls{TV} dynamics are assumed to be subject to uncertainties. In the case of vehicles, tasks are maneuvers. Therefore, we consider maneuver uncertainty and maneuver execution uncertainty. The \gls{TV} dynamics are in the form of~\eqref{eq:dynamics_DO} with $\Mat{A}^\text{DO}$, $\Mat{B}^\text{DO}$, states, and inputs similar to~\eqref{eq:AB} as well as $ \Mat{G}^\text{DO} = \text{diag} \left(0.05, \, 0.067, \, 0.013, \, 0.03\right) $ accounting for diverse \gls{TV} uncertainty in longitudinal and lateral direction. The covariance matrix of the normally distributed \gls{TV} maneuver execution uncertainty $ \bm{w}_{k}^{\text{ex}} \sim \mathcal{N} \left(\bm{0},\Mat{\Sigma}_{k}^{\text{ex}}\right) $ is an identity matrix $ \Mat{\Sigma}_{k}^{\text{ex}} = \text{diag} (1, 1, 1, 1) $. Furthermore, additive measurement noise $ \bm{\nu}_{k} \sim \mathcal{N} \left(\bm{0},\Mat{\Sigma}_{k}^{\nu}\right) $ is considered for $x_{k}^\text{TV}$ and $y_{k}^\text{TV}$ with $\Mat{\Sigma}_{k}^{\nu} =  \text{diag} \lr{0.16, 0.01}$.
The \glspl{TV} have multiple maneuver options with associated maneuver probabilities. The possible maneuvers consist of lane changes to left (LCL) and right (LCR), lane keeping (LK), accelerating (AC), braking (BR), and insignificant acceleration (IA), as well as a combination of the lateral and longitudinal maneuvers, resulting in a total of nine possible maneuvers.

The road consists of three lanes with lane width $ l_{\text{lane}} = \SI{3.5}{\metre} $, where the center of the left lane represents $y = 0$. All vehicles are $ a_{\text{veh}} = \SI{6}{\metre} $ in length and $ b_{\text{veh}} = \SI{2}{\metre} $ in width. The initial lateral position of all vehicles coincides with the lateral center of the vehicles' respective lanes with zero lateral velocity. The initial longitudinal positions and velocities of all vehicles are summarized in Table~\ref{tab:initial}. 
\begin{table}
\vspace{2mm}
\centering
    \caption{Initial Vehicle Configuration}
    \label{tab:initial}
    \begin{tabular}{l c c c c c c}
      \toprule 
       & EV & TV1 & TV2  & TV3 & TV4 & TV5\\
      \midrule 
      \textbf{$x$-pos. $(\si{\metre})$} & 0 & -25 & 25 & 40 & -30 & -10\\
      \textbf{$x$-vel. $(\si{\metre\per\second})$} & 27 & 17 & 27 & 27 & 27 & 22\\
      \bottomrule 
    \end{tabular}
  \vspace{-4mm}
\end{table}
The \gls{TV} reference state is chosen as $ \bm{\xi}_{\text{ref},k}^{\text{TV}} = [0, \, v_{x,\text{ref},k}^{\text{TV}}, \, y_{\text{ref},k}^{\text{TV}}, \, 0 ]^\top $, where $v_{x,\text{ref},k}^{\text{TV}}$ and $y_{\text{ref},k}^{\text{TV}}$ may vary over time depending on the scenario. The feedback controller for the \glspl{TV} is 
\begin{equation}
\Mat{K}^{\text{DO}} = \begin{bmatrix} 0 & -1.0 & 0 & 0 \\ 0 & 0 &-0.8 & -2.2 \end{bmatrix}. 
\end{equation}

To prevent collisions, a region around the \gls{TV} is inadmissible for the \gls{EV}. This is referred to as the safety constraint, where the admissible area is the safe set \gls{xisafe}. In line with \cite{BruedigamEtalWollherr2018b}, we impose a safety constraint modeled as an ellipse. Its definition adheres to
\begin{equation}
\label{eq:dk ellipse}
d_{k} = \frac{\left(\Delta x_{k}\right)^{2}}{a^{2}} + \frac{\left(\Delta y_{k}\right)^{2}}{b^{2}} - 1 \geq 0,
\end{equation}
where we decompose the distance between the \gls{EV} and \gls{TV} into a longitudinal and a lateral component $\Delta x_{k} = x_{k}^{\text{EV}} - \overline{x}_{k}^{\text{TV}}$ and $\Delta y_{k} = y_{k}^{\text{EV}} - \overline{y}_{k}^{\text{TV}}$.

The ellipse center coincides with the \gls{TV} center. Therefore, \eqref{eq:dk ellipse} is fulfilled if the \gls{EV} center lies outside the inner space or on the edge of the ellipse, i.e., $d_k \geq 0$. The parameters $a=30$ and $b=2$ represent the semi-major and semi-minor axis of the ellipse, respectively. 
The values of $a$ and $b$ are chosen conservatively, i.e., the area covered by the safety ellipse is larger than the vehicle shape.

To reduce the number of constraints for sampled \gls{TV} maneuvers, we first introduce a method to adapt the safety constraint ellipse~\eqref{eq:dk ellipse}. As an example, we assume that all possible maneuvers are sampled. Then, as mentioned as a possibility in Section~\ref{sec:sscmpc}, we combine the individual constraint ellipses of all sampled maneuvers at each time step, as shown in Figure~\ref{fig:ellipse constraint combined}. If less maneuvers are sampled, the aggregated ellipse only covers the sampled maneuvers.

\begin{figure}
\vspace{1.5mm}
\begin{center}
\resizebox {1\columnwidth} {!} {
\begin{tikzpicture}

\node at (16.55,-2.2) {\Huge \textbf{\ac{MPC} time step}};

\draw[line width = 0.3mm] (2.2,0) -- (32.9,0);
\draw[line width = 0.3mm] [dashed] (2.2,2) -- (32.9,2);
\draw[line width = 0.3mm] [dashed] (2.2,4) -- (32.9,4);
\draw[line width = 0.3mm] (2.2,6) -- (32.9,6);
\draw[line width = 0.3mm,white] (2.2,6.5) -- (32.9,6.5);

\draw [line width=0.8mm, color = black, ->] (2.2,-0.7) -- (32.9,-0.7);
\draw [line width=0.8mm, color = black] (5,-0.7) -- (5,-0.5);
\node at (5,-1.2) {\Huge $k$};
\draw [line width=0.8mm, color = black] (13,-0.7) -- (13,-0.5);
\node at (13,-1.2) {\Huge $k+1$};
\draw [line width=0.8mm, color = black] (23,-0.7) -- (23,-0.5);
\node at (23,-1.2) {\Huge $k+2$};

\draw[line width = 0.2mm] (3,3) -- (26.6,3);
\draw (3,3) .. controls (11,2.8) and  (18,5.2) .. (26.6,5);
\draw (3,3) .. controls (11,3.2) and  (18,0.8) .. (26.6,1);

\draw (12.7,3) [fill=black!40] circle (0.2cm);
\draw (12.7,3.7) [fill=black!40] circle (0.2cm);
\draw (12.7,2.3) [fill=black!40] circle (0.2cm);
\draw (13.3,3) [fill=black!40] circle (0.2cm);
\draw (13.3,3.77) [fill=black!40] circle (0.2cm);
\draw (13.3,2.23) [fill=black!40] circle (0.2cm);

\draw (5,3)[black,fill=gray!10,fill opacity=0.2,line width = 0.5mm] ellipse (2.5cm and 1.25cm);
\draw [rounded corners = 0.5mm, fill = cyan!40](3.5,2.3) rectangle (6.5,3.7);
\node at (5,3) {\huge TV$^{\text{IA,LK}}_k$};

\draw (23,3)[black,fill=gray!10, fill opacity=0.2,line width = 0.5mm] ellipse (9.1cm and 3.4cm); 

\draw (19.4,1)[black,fill=gray!10, fill opacity=0.2,dotted,line width = 0.5mm] ellipse (2.5cm and 1.25cm);
\draw (23,1)[black,fill=gray!10, fill opacity=0.2,dotted,line width = 0.5mm] ellipse (2.5cm and 1.25cm);
\draw (26.6,1)[black,fill=gray!10, fill opacity=0.2,dotted,line width = 0.5mm] ellipse (2.5cm and 1.25cm);
\draw [rounded corners = 0.5mm, fill = cyan!20](17.9,0.3) rectangle (20.9,1.7);
\draw [rounded corners = 0.5mm, fill = cyan!20](21.5,0.3) rectangle (24.5,1.7);
\draw [rounded corners = 0.5mm, fill = cyan!20](25.1,0.3) rectangle (28.1,1.7);
\node at (19.4,1) {\huge TV$^{\text{BR,LCR}}_{k+2}$};
\node at (23,1) {\huge TV$^{\text{IA,LCR}}_{k+2}$};
\node at (26.6,1) {\huge TV$^{\text{AC,LCR}}_{k+2}$};

\draw (19.4,3)[black,fill=gray!10, fill opacity=0.2,dotted,line width = 0.5mm] ellipse (2.5cm and 1.25cm);
\draw (23,3)[black,fill=gray!10, fill opacity=0.2,dotted,line width = 0.5mm] ellipse (2.5cm and 1.25cm);
\draw (26.6,3)[black,fill=gray!10, fill opacity=0.2,dotted,line width = 0.5mm] ellipse (2.5cm and 1.25cm);
\draw [rounded corners = 0.5mm, fill = cyan!20](17.9,2.3) rectangle (20.9,3.7);
\draw [rounded corners = 0.5mm, fill = cyan!20](21.5,2.3) rectangle (24.5,3.7);
\draw [rounded corners = 0.5mm, fill = cyan!20](25.1,2.3) rectangle (28.1,3.7);
\node at (19.4,3) {\huge TV$^{\text{BR,LK}}_{k+2}$};
\node at (23,3) {\huge TV$^{\text{IA,LK}}_{k+2}$};
\node at (26.6,3) {\huge TV$^{\text{AC,LK}}_{k+2}$};

\draw (19.4,5)[black,fill=gray!10, fill opacity=0.2,dotted,line width = 0.5mm] ellipse (2.5cm and 1.25cm);
\draw (23,5)[black,fill=gray!10, fill opacity=0.2,dotted,line width = 0.5mm] ellipse (2.5cm and 1.25cm);
\draw (26.6,5)[black,fill=gray!10, fill opacity=0.2,dotted,line width = 0.5mm] ellipse (2.5cm and 1.25cm);
\draw [rounded corners = 0.5mm, fill = cyan!20](17.9,4.3) rectangle (20.9,5.7);
\draw [rounded corners = 0.5mm, fill = cyan!20](21.5,4.3) rectangle (24.5,5.7);
\draw [rounded corners = 0.5mm, fill = cyan!20](25.1,4.3) rectangle (28.1,5.7);
\node at (19.4,5) {\huge TV$^{\text{BR,LCL}}_{k+2}$};
\node at (23,5) {\huge TV$^{\text{IA,LCL}}_{k+2}$};
\node at (26.6,5) {\huge TV$^{\text{AC,LCL}}_{k+2}$};

\end{tikzpicture}
}
\caption{Qualitative depiction of the combined safety constraint ellipse. Safety ellipses for step $k+1$ omitted. } 
\label{fig:ellipse constraint combined}
\end{center}
\vspace{-4.5mm}
\end{figure}
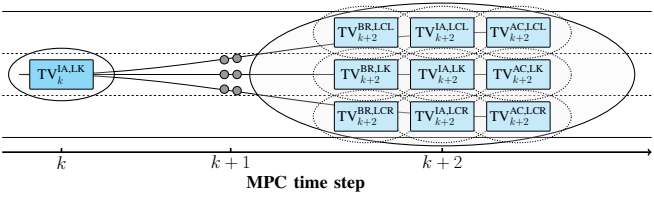

The result is the aggregated ellipse
\begin{subequations}
  \label{eq:ellipse combined}
  \begin{align}
    \label{eq:dk ellipse combined}
    \tilde{d}_k &= \frac{\left( \Delta \tilde{x}_k \right)^{2}}{\tilde{a}_k^{2}} + \frac{\left( \Delta \tilde{y}_k \right)^{2}}{\tilde{b}_k^{2}} - 1 \geq 0, \\
    \label{eq:delta x combined}
    \Delta \tilde{x}_k &= x_k - \tilde{x}_k^{\text{TV}}, \\
    \label{eq:delta y combined}
    \Delta \tilde{y}_k &= y_k - \tilde{y}_k^{\text{TV}}, \\
    \label{eq:ellipse combined x}
    \tilde{x}_k^{\text{TV}} &= \frac{x_k^{\text{TV,IA}} + x_k^{\text{TV,BR}} + x_k^{\text{TV,AC}}}{3}, \\
    \label{eq:ellipse combined y}
    \tilde{y}_k^{\text{TV}} &= \frac{y_k^{\text{TV,LK}} + y_k^{\text{TV,LCL}} + y_k^{\text{TV,LCR}}}{3}
  \end{align}
\end{subequations}
with center $\left( \tilde{x}_k^{\text{TV}}, \tilde{y}_k^{\text{TV}} \right)$. The longitudinal and lateral position of the \gls{TV} corresponding to the respective maneuvers are indicated by the variables $x_k^{\text{TV},M}$ and $y_k^{\text{TV},M}, \, M \in \left\{ \text{IA}, \text{BR}, \text{AC}, \text{LK}, \text{LCL}, \\ \text{LCR} \right\}$, respectively. The combined ellipse exhibits the adjusted semi-major and semi-minor axes
\begin{subequations}
  \label{eq:dk ellipse combined axes}
  \begin{align}
    \label{eq:dk ellipse combined a}
    \tilde{a}_k &= a + 0.5 \left| x_k^{\text{TV, AC}} - x_k^{\text{TV, BR}} \right| + \frac{2}{l_{\text{lane}}} \left( \tilde{b}_k - b \right), \\
    \label{eq:dk ellipse combined b}
    \tilde{b}_k &= b + 0.5\left| y_k^{\text{TV, LCL}} - y_k^{\text{TV, LCR}} \right|.
  \end{align}
\end{subequations}
By generating the aggregated safety ellipse, the number of necessary constraints is reduced. As seen in Figure~\ref{fig:ellipse constraint combined}, the aggregated ellipse does not necessarily cover all individual safety ellipses perfectly, which is still reasonable as the individual safety ellipses are designed larger than necessary. 

For the \gls{MPC} \gls{ocp}~\eqref{eq:ocp_sscmpc}, a prediction horizon $N=12$ is selected. The cost function terms are set to $l = \left\| \Delta \bm{\xi}_{k} \right\| ^2_{\Mat{Q}} \, + \, \left\| {\bm{u}}_{k} \right\| ^2_{\Mat{R}}$, $J_\text{f} = \left\| \Delta \bm{\xi}_{N} \right\| ^2_{\Mat{S}}$, with the cost function weights $\Mat{Q}, \Mat{S} \in \mathbb{R}^{4 \times 4}$, and $\Mat{R} \in \mathbb{R}^{2 \times 2}$, as well as $\left\| \bm{z} \right\| ^2_{\Mat{Z}} = \bm{z}^\top \Mat{Z} \bm{z}$ and $\Delta \bm{\xi}_{k} = \bm{\xi}_k - \bm{\xi}_{\text{ref},k}$ with reference $\bm{\xi}_{\text{ref},k}$. For positional reference tracking in $y$-direction, the \gls{EV} reference is set to its current lane center, while $v_{y,\text{ref},k} = 0$ and $v_{x, \text{ref},k} = \SI{27}{\metre\per\second}$. For the cost function, the first element of $\Delta \bm{\xi}_k$ is neglected, since no reference for $x_k$ is imposed. Here, $\Mat{Q} = \Mat{S} = \text{diag}(0, 3, 0.5, 0.1)$, $\Mat{R} = \text{diag}(1, 0.1)$ are selected.

While the maneuver probabilities are scenario specific and different task uncertainty risk parameters \gls{betata} are evaluated, the task execution risk parameter is chosen to be $\gls{betaex} = 0.8$. In case the original \gls{MPC} problem is infeasible, a recovery \gls{MPC} \gls{ocp} is solved with slack variables to soften the constraints, as described in~\cite{BruedigamEtalWollherr2018b}. For the recovery \gls{ocp}, the slack variable weight in the cost function is $\lambda = 50$ and the task execution risk parameter is changed to a more conservative value $\beta_{\lambda}^\text{ex} = 0.995$, to prioritize safety. In case the recovery problem fails, the solver selects the last feasible point as the solution to the \gls{ocp}.

Apart from the safety constraint, the \gls{EV} plans its motion subject to the constraints $-1.75 \leq y_{k} \leq 8.75$, $-5 \leq u_{x,k} \leq  5$, $-0.5 \leq u_{y,k} \leq 0.5$, $-1 \leq \Delta u_{x,k} \leq 1 $, $-0.2 \leq \Delta u_{y,k} \leq 0.2$ with $\Delta u_{x,k} = u_{x,k}-u_{x,k-1}$, $\Delta u_{y,k} = u_{y,k}-u_{y,k-1}$.

\subsection{Simulation Results}
\label{sec:simulation results}

In the following, the S+SC MPC algorithm is evaluated in the presented scenario. As mentioned, the maneuver risk parameter \gls{betata} is varied, resulting in a varying sample size $K$. Monte Carlo simulations are conducted 150 times for each risk parameter value. 

Each simulation consists of two parts. For the first 20 steps, the \gls{EV} follows a conservative behavior with $\gls{betata} = 0.999$, representing a behavior prediction initialization phase. The \gls{EV} assumes that the \gls{TV} probabilities for lane changes or changes in acceleration are $p^\text{LC} = 0.80$ and $p^\text{AC} = p^\text{BR} = 0.40$. In case a lane changes is possible to the left or right, $p^\text{LC}$ is assigned equally. In the second part from step 21 to step 100, it is assumed that the \gls{EV} has adapted its behavior prediction. Therefore, the probabilities of \gls{TV} maneuvers change to $p^\text{LC} = 0.20$ and $p^\text{AC} = p^\text{BR} = 0.10$. For the second part of the simulation, different risk parameters $\gls{betata} \in \left\{ 0.99, 0.95, 0.89, 0.83 \right\}$ are evaluated. Within the actual simulation, all \glspl{TV} maintain their respective lanes, except TV4, which moves to the center lane. The reference velocities in $x$-direction are $v_{x,\text{ref}}^\text{TV1} = \SI{22}{\metre\per\second}$, $v_{x,\text{ref}}^\text{TV2} = \SI{22}{\metre\per\second}$, $v_{x,\text{ref}}^\text{TV3} = \SI{17}{\metre\per\second}$, $v_{x,\text{ref}}^\text{TV4} = \SI{17}{\metre\per\second}$, $v_{x,\text{ref}}^\text{TV5} = \SI{27}{\metre\per\second}$.

The result of an individual example with $\gls{betata} = 0.95$ is illustrated in Figure~\ref{fig:examplesimu}. 
\begin{figure}
\vspace{1mm}
\centering
\includegraphics[width = 0.95\columnwidth]{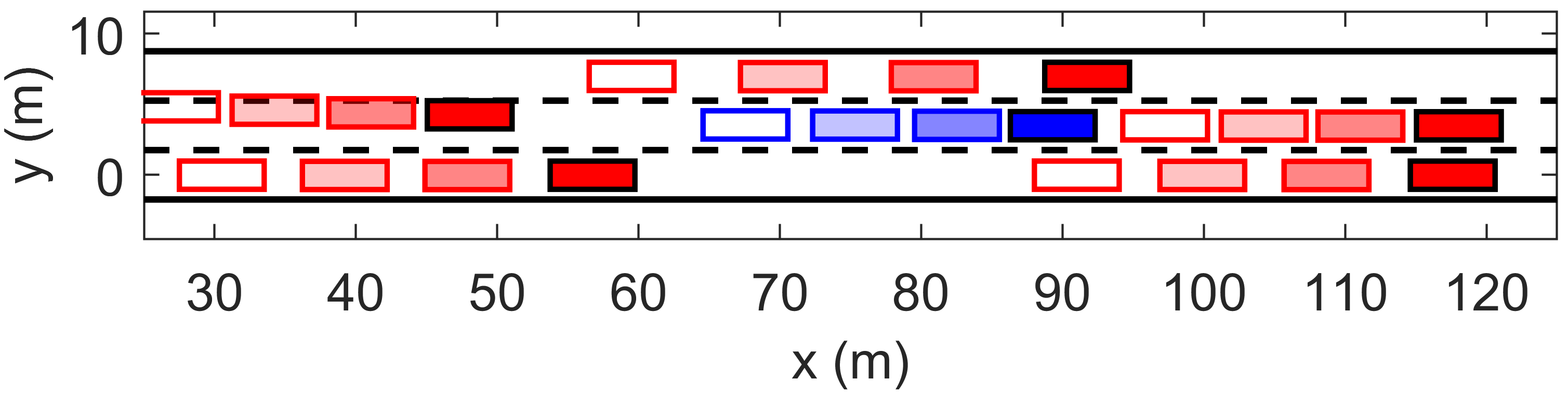} 
\includegraphics[width = 0.95\columnwidth]{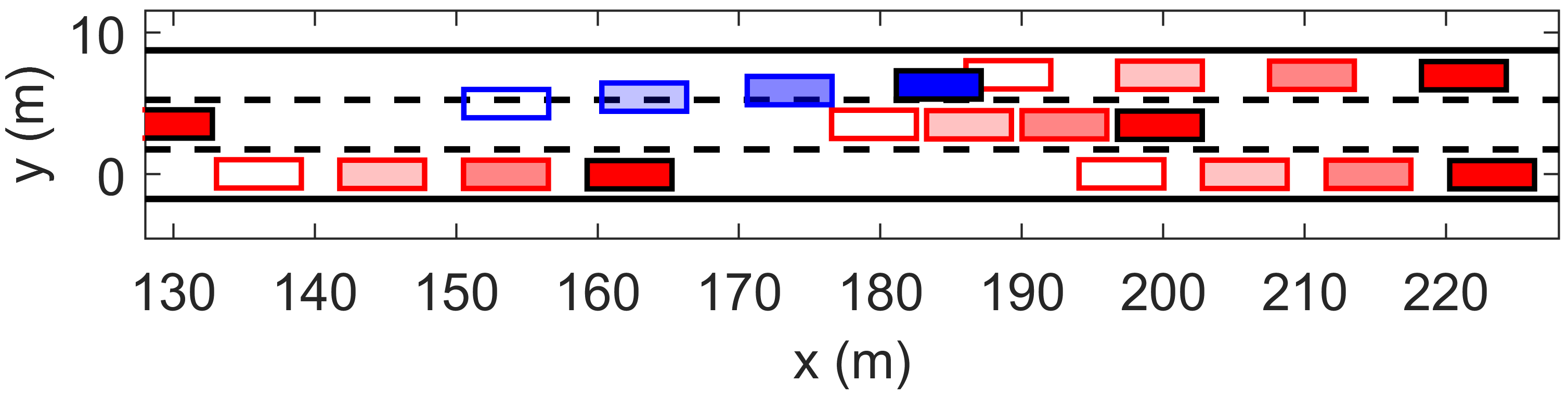}
\caption{Vehicle motion for simulation steps $21$ (top) and $45$ (bottom). The EV is shown in red, TVs in blue. Fading boxes represent past states.}
\label{fig:examplesimu}
\vspace{-2.5mm}
\end{figure}
While there initially is a gap between TV3 and TV5, the \gls{EV} does not plan to overtake, as a potential lane change of either TV3 or TV5 would result in an inevitable collision. Therefore, the \gls{EV} slows down such that TV5 passes TV3 first. Subsequently, the \gls{EV} safely moves to the left lane to overtake TV3.

Even though \gls{SMPC}, in general, allows a small probability of constraint violation, in regular scenarios collisions are avoided as the repetitively updated \gls{SMPC} inputs allow to constantly adjust. For example, it may not be possible to satisfy the chance constraint for a late prediction step within the \gls{SMPC} horizon, due to an unexpected uncertainty realization. The \gls{ocp} is therefore infeasible. However, a collision may still be prevented in the next steps, depending on the future uncertainty realizations. Here, we designed a challenging situation for the \gls{EV}, as lane changes are considered to be probable for all \glspl{TV} and must be accounted for. The results of the Monte Carlo simulations are shown in Table~\ref{tab:results_sscmpc}. 
\begin{table}
\vspace{2mm}
\centering
    \caption{S+SC MPC Simulation Results}
    \label{tab:results_sscmpc}
    \begin{tabular}{l c c c c}
      \toprule 
       \textbf{risk parameter} \gls{betata}  & 0.99 & 0.95 & 0.89  & 0.83 \\
      \midrule 
      \textbf{collisions} & 0 & 0 & 0 & 0 \\
      \addlinespace
      \textbf{cost $J_{100}$} & $3.64 \mathrm{e}4$ & $3.40 \mathrm{e}4$ & $3.59 \mathrm{e}4$ & $3.76 \mathrm{e}4$\\
      \addlinespace
      \makecell[l]{\textbf{infeasible OCP} \\ \textbf{steps}} & $26.3$ & $25.2$ & $24.2$ & $26.6$ \\
      \addlinespace
      \makecell[l]{\textbf{infeasible rec.} \\ \textbf{OCP steps}} & $2.2$ & $3.2$ & $5.2$ & $7.4$ \\
      \bottomrule
    \end{tabular}
  \vspace{-4mm}
\end{table}

Summarizing the simulation results, the first important observation is that no collisions occurred. While the safety ellipse is slightly violated in some simulation runs, the safety ellipse is chosen large enough that no collisions followed. 

The performance is evaluated by computing the cost at each time step, based on the actual states and inputs, with
\begin{IEEEeqnarray}{c}
J_{100} = \sum_{k=0}^{99} \left\| \Delta \bm{\xi}_{k+1} \right\| ^2_{\Mat{Q}} \, + \, \left\| {\bm{u}}_{k} \right\| ^2_{\Mat{R}}.
\end{IEEEeqnarray}
The cost remains on a similar level for all risk parameters, where the best choice in this scenario is $\gls{betata}= 0.95$. Lower risk increases conservatism, while high risk results in less smooth control inputs, again increasing the cost.

As mentioned before, the potential lane changes of all \glspl{TV} pose a challenging situation for the \gls{EV}, resulting in steps where the \gls{ocp} becomes infeasible. However, the steps with successfully solved recovery \glspl{ocp} are significantly more likely, especially for a low accepted level of risk. The average computation time is $\SI{214}{\ms}$.

\subsection{Comparison to SMPC and SCMPC}
\label{sec:comparison}

We now compare the results of S+SC MPC to only applying SMPC or SCMPC. the results are shown in Table~\ref{tab:results_smpc_scmpc}.
\begin{table}
\vspace{2mm}
\centering
    \caption{SMPC and SCMPC Simulation Results}
    \label{tab:results_smpc_scmpc}
    \begin{tabular}{l c | c c c c}
      \toprule
      & SMPC & \multicolumn{4}{c}{SCMPC} \\
      \textbf{risk parameter} & 0.8 & 0.99 & 0.95 & 0.89  & 0.83 \\
      \midrule 
      \textbf{collisions} & 79 & 49 & 43 & 45 & 41 \\
      \addlinespace
      \textbf{cost $J_{100}$} & $3.22 \mathrm{e}4$ & $6.77 \mathrm{e}4$ & $6.27 \mathrm{e}4$ & $6.88 \mathrm{e}4$ & $7.08 \mathrm{e}4$\\
      \addlinespace
      \makecell[l]{\textbf{infeasible OCP} \\ \textbf{steps}} & $31.2$ & $54.1$ & $53.6$ & $55.7$ & $54.7$ \\
      \addlinespace
      \makecell[l]{\textbf{infeasible rec.} \\ \textbf{OCP steps}} & $21.8$ & $33.8$ & $33.6$ & $34.1$ & $34.3$ \\
      \bottomrule 
    \end{tabular}
  \vspace{-4mm}
\end{table}
First, an analytic SMPC algorithm, inspired by \cite{CarvalhoEtalBorrelli2014}, is analyzed with $\gls{betaex} = 0.8$. The advantage of S+SC MPC is that the mixed uncertainty structure is exploited. Applying only SMPC, in order to account for maneuver and execution uncertainty, multiple possible maneuvers would need to be approximated by a Gaussian uncertainty. However, this would result in a major increase of the safety ellipse, covering the entire road width, rendering overtaking other \glspl{TV} impossible. Therefore, in the SMPC simulation, the SMPC algorithm only accounts for maneuver execution uncertainty.

A total of 79 collisions occurred. While the cost is slightly lower compared to S+SC MPC, significantly more steps with infeasible \glspl{ocp} occur, especially for the recovery problem.

In the SCMPC simulation, inspired by~\cite{SchildbachBorrelli2015}, the maneuver execution uncertainty is approximated by samples. To compare a similar situation as in the SMPC simulation, no task uncertainty is considered here. Again, a significant number of simulation runs result in collisions, while the cost also increases compared to S+SC MPC. The steps with infeasible \glspl{ocp} appear more often than in the S+SC MPC simulation runs. The computation times for SMPC and SCMPC are similar to S+SC MPC.

\subsection{Varying Vehicle Settings}

So far, only one vehicle setting is considered. Therefore, we additionally ran 150 simulations with randomly chosen \gls{TV} settings for each simulation run (similar initial EV state as before). The \glspl{TV} get assigned random initial positions $x^\text{TV}_0 \in [-150, 150]$ and are placed on one of the three lanes, i.e., $y_0 \in \{0, 3.5, 7\}$. The constant longitudinal velocity for each \gls{TV} is randomly chosen according to $v_x^\text{TV} \in [17, 27]$ with $v_y^\text{TV} = 0$. It is ensured that all vehicles positioned on similar lanes have enough longitudinal distance $\Delta x \geq 50$, and velocities are selected such that \gls{TV} collisions are avoided. The proposed S+SC MPC method successfully handled all 150 simulation runs and no collisions occurred.

Overall, S+SC MPC allows exploiting the uncertainty structure of the simulation setting, achieving adequate performance and avoiding collisions. While the results presented here are promising, it is to note that the benefits of the proposed method depend on the application setting and to which degree the uncertainty structure may be exploited.


\section{Conclusion}
\label{sec:conclusion}

The proposed S+SC MPC method allows considering the specific uncertainty structure found in many applications, where both task uncertainty and task execution uncertainty are present. As \gls{SCMPC} is suitable for non-Gaussian task uncertainty and \gls{SMPC} copes well with Gaussian execution uncertainty, the combination shows promising results.

While in this work the S+SC MPC method is applied to a vehicle scenario, the framework is designed in a general way, such that it is applicable also to other applications, e.g., human-robot collaboration. In this robotics setting, a robotic arm may have the option of moving to one of several items, while the exact motion towards the specific item may vary.  Without specifically focusing on agents, the proposed framework may also be applicable to process control or finance.

\appendices

\bibliography{./references/Dissertation_bib}

\begin{thebibliography}{10}

\bibitem{LangsonEtalMayne2004}
W.~Langson, I.~Chryssochoos, S.V. Raković, and D.Q. Mayne.
\newblock Robust model predictive control using tubes.
\newblock {\em Automatica}, 40(1):125 -- 133, 2004.

\bibitem{RawlingsMayneDiehl2017}
J.B. Rawlings, D.Q. Mayne, and M.~Diehl.
\newblock {\em Model Predictive Control: Theory, Computation, and Design}.
\newblock Nob Hill Publishing, 2017.

\bibitem{Mesbah2016}
A.~Mesbah.
\newblock Stochastic model predictive control: An overview and perspectives for
  future research.
\newblock {\em IEEE Control Systems}, 36(6):30--44, Dec 2016.

\bibitem{FarinaGiulioniScattolini2016}
M.~Farina, L.~Giulioni, and R.~Scattolini.
\newblock Stochastic linear model predictive control with chance constraints
  – a review.
\newblock {\em Journal of Process Control}, 44(Supplement C):53 -- 67, 2016.

\bibitem{SchwarmNikolaou1999}
A.T. Schwarm and M.~Nikolaou.
\newblock Chance-constrained model predictive control.
\newblock {\em AIChE Journal}, 45(8):1743--1752, 1999.

\bibitem{KouvaritakisEtalCheng2010}
B.~Kouvaritakis, M.~Cannon, S.V. Rakovic, and Q.~Cheng.
\newblock Explicit use of probabilistic distributions in linear predictive
  control.
\newblock {\em Automatica}, 46(10):1719 -- 1724, 2010.

\bibitem{CarvalhoEtalBorrelli2014}
A.~Carvalho, Y.~Gao, S.~Lefevre, and F.~Borrelli.
\newblock Stochastic predictive control of autonomous vehicles in uncertain
  environments.
\newblock In {\em 12th International Symposium on Advanced Vehicle Control},
  Tokyo, Japan, 2014.

\bibitem{BlackmoreEtalWilliams2010}
L.~Blackmore, M.~Ono, A.~Bektassov, and B.C. Williams.
\newblock A probabilistic particle-control approximation of chance-constrained
  stochastic predictive control.
\newblock {\em Trans. Rob.}, 26(3):502--517, June 2010.

\bibitem{SchildbachEtalMorari2014}
G.~Schildbach, L.~Fagiano, C.~Frei, and M.~Morari.
\newblock The scenario approach for stochastic model predictive control with
  bounds on closed-loop constraint violations.
\newblock {\em Automatica}, 50(12):3009 -- 3018, 2014.

\bibitem{BruedigamEtalWollherr2018b}
T.~Br\"udigam, M.~Olbrich, M.~Leibold, and D.~Wollherr.
\newblock Combining stochastic and scenario model predictive control to handle
  target vehicle uncertainty in autonomous driving.
\newblock In {\em 21st IEEE International Conference on Intelligent
  Transportation Systems (ITSC)}, 2018.

\bibitem{SchildbachBorrelli2015}
G.~Schildbach and F.~Borrelli.
\newblock Scenario model predictive control for lane change assistance on
  highways.
\newblock In {\em 2015 IEEE Intelligent Vehicles Symposium (IV)}, pages
  611--616, Seoul, South Korea, June 2015.

\bibitem{CesariEtalBorrelli2017}
G.~Cesari, G.~Schildbach, A.~Carvalho, and F.~Borrelli.
\newblock Scenario model predictive control for lane change assistance and
  autonomous driving on highways.
\newblock {\em IEEE Intelligent Transportation Systems Magazine}, 9(3):23--35,
  Fall 2017.

\bibitem{BruedigamEtalLeibold2020c}
T.~Br\"udigam, F.~di~Luzio, L.~Pallottino, D.~Wollherr, and M.~Leibold.
\newblock Grid-based stochastic model predictive control for trajectory
  planning in uncertain environments.
\newblock In {\em 23rd IEEE International Conference on Intelligent
  Transportation Systems (ITSC)}, 2020.

\bibitem{MuraleedharanEtalSuzuki2020}
A.~Muraleedharan, A.~Tran, H.~Okuda, and T.~Suzuki.
\newblock Grid-based stochastic model predictive control for trajectory
  planning in uncertain environments.
\newblock In {\em IFAC World Congress 2020}, Berlin, Germany, 2020.

\bibitem{SchulmanEtalAbbeel2013}
J.~Schulman, J.~Ho, A.~Lee, I.~Awwal, H.~Bradlow, and P.~Abbeel.
\newblock Finding locally optimal, collision-free trajectories with sequential
  convex optimization.
\newblock In {\em Robotics: Science and Systems 2013}, 2013.

\bibitem{BruedigamEtalLeibold2021b}
T.~Br{\" u}digam, M.~Olbrich, D.~Wollherr, and M.~Leibold.
\newblock Stochastic model predictive control with a safety guarantee for
  automated driving.
\newblock {\em IEEE Transactions on Intelligent Vehicles}, pages 1--1, 2021.

\bibitem{GiftthalerEtalBuchli2018}
M.~Giftthaler, M.~Neunert, M.~{St\"auble}, and J.~Buchli.
\newblock The {Control Toolbox} - an open-source {C++} library for robotics,
  optimal and model predictive control.
\newblock In {\em 2018 IEEE International Conference on Simulation, Modeling,
  and Programming for Autonomous Robots (SIMPAR)}, pages 123--129, May 2018.

\end{thebibliography}
\bibliographystyle{unsrt}

\end{document}